\documentclass[a4paper,11pt,twoside]{amsart} 
\usepackage[T1]{fontenc}
\usepackage[utf8]{inputenc}
\usepackage[english]{babel} % set input encoding (not needed with XeLaTeX)
\usepackage { amsfonts }
\usepackage { amsmath }
\usepackage { amsthm }
\usepackage{amscd}
\usepackage{mathtools}
\usepackage{amssymb}
\usepackage{hyperref}
\usepackage{xcolor}
\usepackage{enumitem}
\usepackage{caption}
\usepackage{tikz-cd}
\usepackage[top=4cm,bottom=4cm,left=3.5cm,right=3.5cm]{geometry}
\usepackage{longtable}
\usepackage{tabularx}

\usepackage[numbers]{natbib}
\usepackage{subcaption}

\usepackage{tabularx}

\newcommand{\KK}{\mathbb{K}}

\newcommand{\I}{\mathcal{I}}

\newcommand{\mo}[1]{\lvert  #1 \rvert}
\newcommand{\mat}[2]{#1_{[#2]} }
\newcommand{\nmat}[3]{N_{#1}(#2,#3) }

\newcommand{\Z}{\mathbb{Z}}

\theoremstyle{definition}
\newtheorem{definizione}{Definition}[section]
\theoremstyle{definition}
\newtheorem{teorema}[definizione]{Theorem}
\theoremstyle{definition}
\newtheorem{proposizione}[definizione]{Proposition}
\theoremstyle{definition}
\newtheorem{lemma}[definizione]{Lemma}
\theoremstyle{definition}
\newtheorem{ex}[definizione]{Example}
\theoremstyle{definition}
\newtheorem{rem}[definizione]{Remark}
\theoremstyle{definition}
\newtheorem{nota}[definizione]{Notation}
\theoremstyle{definition}
\newtheorem{cor}[definizione]{Corollary}

\theoremstyle{definition}
\numberwithin{equation}{section}

\title{The weight distribution of codes over finite chain rings}

\begin{document}

\author[G. Cavicchioni]{Giulia Cavicchioni}
\address{Department of Mathematics\\
University of Trento \\
Italy
}
\email{giulia.cavicchioni@unitn.it}

\author[A. Meneghetti]{Alessio Meneghetti}
\address{Department of Mathematics\\
University of Trento \\
Italy
}
\email{alessio.meneghetti@unitn.it}



\maketitle
 \begin{abstract}\footnotesize
 In this work, we determine new linear equations for the weight distribution  of linear codes over finite chain rings. The identities are determined  by  counting  the number of some special submatrices of the parity-check matrix of the code.
Thanks to these relations we  are able to compute the full weight distribution of codes with small Singleton defects, such as MDS, MDR and AMDR codes. 
\end{abstract}
\section{Introduction}\label{intro}
Ring-linear coding theory has been widely studied because of its theoretical and practical interest. On one hand, ring-linear codes are relevant from an algebraic perspective: as shown in  \cite{kerdock}, some optimal but non-linear binary codes can be represented as linear codes over $\Z/4\Z$ endowed with the Lee metric. On the other hand, ring-linear codes have received attention in cryptographic community. The recent effort among cryptographers to obtain secure post-quantum ciphers \cite{bike,crystals,bernstein,mceliece,hqc,nist} led to an increase in the interest in computationally hard algebraic problems, and an interested reader can refer to \cite{berlekamp,meneghetti2022,vardy}  for more details. Code-based cryptography is one of the most studied and promising areas in the  post-quantum framework. However, due to the necessity of reducing the public key size associated to a code-based cryptosystem, there has been interest in exploring different ambient spaces and metrics other then vector spaces over finite fields equipped with the Hamming metric. For example, codes over finite rings equipped with the Lee metric may decrease the-public key size of the cryptosystems; for further details see \cite{horlemann, persichetti,weger}.\\  Understanding the minimum distance of a code is computationally hard and it is one of the main problem in Coding Theory. In 1997 Vardy proved that, given a basis of a code, determining precisely the minimum distance of a linear code is NP-hard \cite{vardy}. Hence this computational problem, as well as several related questions, is  linked to the security of post-quantum cryptographic protocols.  \\
Even the problem of calculating the weight distribution of a linear code, which implies the determination of the minimum distance, is NP-hard. In this paper we treat the problem of computing the  weight distribution of linear codes over finite chain rings equipped with the Hamming metric. 
In classical coding theory the most fundamental result about weight distributions are the MacWilliams identities, which express how the weight enumerators of a linear code and its dual relate to each other. Several authors
have generalized this work in  different directions. For example, a MacWilliams theorem for codes over finite Frobenius rings was given by Wood in 1999 \cite{wood99}.\\
Here, we provide new linear equations for the weight distribution of ring-linear codes  by  counting  the number of some special submatrices of the parity-check matrix of the code. This task is  certainly  as difficult as the original one; however, it allows to investigate codes having special structure in their parity-check matrix. The  provided equations and MacWilliams identities  seems to be independent, but there could be a possible link in between this equations and some variant of MacWilliams identities. \\
This paper is organized as follows. In Section \ref{preliminiries} we recall some basics on  linear codes over finite fields. In Section \ref{GeneralitiesLinC} we introduce ring-linear codes; we investigate the structure of the parity-check matrix and the  weight distribution of a linear code. % In Section \ref{Freewd} we restrict to the case of free linear codes. 
In Section \ref{Linearwd} we derive new relations for the weight distribution of ring-linear codes; we discuss the optimality of the result. The obtained  formula is a modification of the formula given in   \cite[Proposition 5]{meneghetti2021} for  linear codes over finite fields and specialised in \cite{pellegrini} for Hermitian codes. In Section \ref{appwdr} we apply our  formula to verify the known results about the distribution of MDS codes. Moreover we derive the weight distribution formula for MDR and AMDR codes.  Finally, in Section  \ref{Relation with MacWilliams identities}, we discuss the connection between MacWilliams identities and the provided relations.
%The determination of the weight distribution $A(C)=(A_1,\dots,A_n)$ of a code $C$ is of both of  theoretical and practical interest. A question naturally arising in the context of algebraic codes is whether we can take advantage of the underlying algebraic structure to obtain some information useful to the computation of the weight distribution.

%However, in literature there has been an effort in studying the weight distribution of codes that have rich algebraic structure or that are optimal with respect to some Singleton-like bound. 

\section{Preliminaries on Linear codes over finite fields}\label{preliminiries}
In its most general setting, Coding Theory is the study of discrete sets equipped with a metric. The most studied case is that of algebraic varieties living in vector spaces over finite fields, and the metric is the Hamming metric. In this framework, a (linear) code $C$ is a vector subspace of dimension $k$ of $\left(\mathbb{F}_q\right)^n$, where the elements of the code are called codewords and the parameters $n$ and $k$ are respectively known as the length and the dimension of $C$. The Hamming metric, also known as Hamming distance, is a discrete metric counting the number of non-zero coordinates, namely,
$$
\mathrm{d}(v,w)=\left|\left\{i\mid v_i\neq w_i, \;1\leq i\leq n\right\}\right|\;,
$$
for any $v=(v_1, \ldots ,v_n)$ and $w=(w_1,\ldots,w_n)$ in $\left(\mathbb{F}_q\right)^n$.
\\
Notice that in this work we consider the elements of vector spaces and modules to be row vectors, a standard notation in Coding Theory. If $v$ is any (row) vector, then its transpose $v^{\top}$ is a column vector.
\\
The third most important parameter of a code is the so-called minimum distance $d$, which is the minimum among the Hamming distances of any pair of distinct codewords, i.e.
$$
d=\min_{c_1,c_2\in C}\mathrm{d}(c_1,c_2)\;,
$$ and it coincides with the minimum weight of a codeword. 
The importance of the minimum distance is related to the capability of codes to correct errors. If we are presented with a vector $v$ that should be a codeword $c$ of a given code $C$, even if its coordinates are corrupted (hence $v\notin C$), then we can safely reconstruct $c$ from $v$ provided that the number of erroneous coordinates of $v$ is bounded by $\lfloor\frac{d-1}{2}\rfloor$. % Understanding the minimum distance of a code is computationally hard and one of the main problem in Coding Theory. In 1997 Vardy proved that, given a basis of a code, determining precisely the minimum distance of a linear code is NP-hard \cite{vardy}. This computational problem, as well as several related questions arising from Algebraic Coding Theory, are linked to the security of Post-Quantum Cryptographic Protocols, a branch of Cryptography devoted to the study of ciphers capable of resisting against quantum adversaries. The recent effort among cryptographers to obtain secure post-quantum ciphers led to an increase in the interest in computationally hard algebraic problems, and an interested reader can refer to \cite{} for more details. As in the classical case, the Hamming weight distribution can be encoded in the coefficient of a polynomial,  the so called weight enumerator polynomial. 
\\The weight distribution of a code of length $n$ specifies the number of codewords of each possible weight $0, 1, \dots, n$. Even if  the weight distribution
does not in general uniquely determine a code, it gives important information: 
in addition to providing the correction capability of a  code,  it allows  to calculate the probability of undetected errors  (see \cite[Chapter 2]{torleiv}). 
\section{Generalities on linear codes over finite chain rings}\label{GeneralitiesLinC}

	A finite ring with unity $ 1\neq0 $ is called a left (resp. right) chain ring if its left (resp. right) ideals are linearly ordered by inclusion. 
Note that a finite chain ring is a  local ring where  all the ideals are principal.\\ Throughout the paper let $R $ be a finite commutative chain ring. Let $ \gamma $ be the  generator of the maximal ideal and let $s$ be its the nilpotency index.   Let $ \KK $ denote the residue field with $p$ elements $ R/\gamma R $. 

\begin{definizione}  A  \emph{linear code of length $ n $ in the alphabet $ R $} is a submodule of $ R^n $. The free module $R^n$ is called the \emph{ambient space} of the code.
\end{definizione}  
\begin{definizione}
The\emph{ Hamming weight} of an element $ c=(c_1,\dots,c_n)\in R^n $ is the number $ w(c) $ of non-zero entries of $ c $. 
\end{definizione}

% \section{Matrices over rings} Let $ R $ be a ring and $ M\in M_{m\times n}(R) $. For each $ t=1,\dots,\min\{m,n\} $, $ I_t(M)$ is an ideal  of $ R $ generated by all the $ t\times t $ minors of $ M $. Laplace's theorem ensures that \[ I_r(M)\subseteq I_{r-1}(M)\subseteq\cdots\subseteq I_2(M)\subseteq I_1(M)\subseteq R.  \] Extending the definition of $ I_t(M) $ as \begin{equation*}I_t(M)= 
%\begin{cases}
% (0) \quad \text{if } t>\min\{m,n\}\\\hspace{2pt} R \quad\hspace{5pt}\text{if } t\le 0
% \end{cases}
% \end{equation*}
% we get \begin{equation*}
% 	(0)=I_{r+1}(M)\subseteq I_r(M)\subseteq\cdots\subseteq I_1(M)\subseteq I_0(M)=R
% \end{equation*}
% Computing the annihilator of each ideal in the previous chain we get the following ascending chain of ideals \begin{equation*}
% 	(0)=\ann{R}{R}\subseteq\ann{R}{I_1(M)}\subseteq\cdots\subseteq\ann{R}{I_r(M)}\subseteq\ann{R}{(0)}=R
% \end{equation*} 
% \begin{definizione}
% Let $ M\in M_{m\times n} (R)$. The \emph{MacCoy rank} (or simply \emph{rank}) of $ M $,  $ \rk{M} $, is defined as \[ \rk{M}=\max\{t \mid\ann{R}{I_t(M)}=(0) \}. \]
% \end{definizione}

\begin{definizione}
A matrix $ G $ is called a \emph{generator matrix} for the code  $ C $ over $ R $ if the rows of $ G $ span $ C $ and none of them can be written as a linear combination of the other rows of $ G $. 
\end{definizione}

As shown in \cite{norton}, any linear code over a finite chain ring has a generator matrix. In our framework it is convenient to work with a generator matrix in standard form.
\begin{proposizione}{\cite[Proposition 3.2]{norton}}\label{gmstand}
Let $C$ be a linear code in $R^n$. $C$ is permutation equivalent to a code having the following  generator matrix in standard form:
	\begin{equation*}\label{genma}
G=\begin{bmatrix}
I_{k_0}&A_{0,1} &A_{0,2} &A_{0,3}&\dots& A_{0,s-1}& A_{0,s} \\
0 &\gamma I_{k_1} & \gamma A_{1,2}& \gamma A_{1,3}& \dots& \gamma A_{1,s-1}&\gamma A_{1,s}\\
0 &0 & \gamma^2 I_{k_2} & \gamma^2 A_{2,3}& \dots& \gamma^2 A_{2,s-1}&\gamma^2 A_{2,s}\\

\vdots &  \vdots& \vdots & \vdots && \vdots &\vdots \\
0 &0 & 0 & 0 & \dots & \gamma^{s-1}I_{k_{s-1}}& \gamma^{s-1} A_{s-1,s}\\
\end{bmatrix}\ ,
\end{equation*} where $A_{i,s}\in M_{k_i\times n-K}(R/\gamma^{s-i}R)$ and $A_{i,j}\in M_{k_i\times k_j}(R/\gamma^{s-i}R)$ for $j< s$.\end{proposizione} 
For all $0\le i \le s$ the $ k_i $'s denote the number of rows of $ G $ that are divisible by $ \gamma^i $ but not by $ \gamma^{i+1} $. The  parameters  $k_0,\dots, k_{s-1}$ are the same for all generator matrices in systematic form, and  $ C $ is said to be of \emph{type} $(k_0,k_1,\dots,k_{s-1}) $.  The \emph{rank} of $C$ is defined as $K=\sum_{i=0}^{s-1}k_i$.

\begin{definizione}
The \emph{free rank} of $C$ is defined  to be the maximum of the rank of the free submodules of $C$ and it coincides with $k_0$.\end{definizione}\begin{definizione} A linear code is said to be \emph{free} if its rank coincides with its free rank. In this case, the code is a free $R$-submodule which is isomorphic to $R^{k_0}$. %Therefore, the code is of type $(k_0,0,\dots,0)$ $C$  and we will simply say that it is of \emph{type $k_0$}.
\end{definizione}

If $C$ is a free code, then any  systematic generator matrix has the form \begin{equation*}G=
    \begin{bmatrix}
    I_{k_0} & A 
    \end{bmatrix}\in (\Z/p^s\Z)^{k_0\times n}.
\end{equation*} 

Since for all $0\le j\le s-1$  we have $  \mo{\gamma^jR}=p^{s-j}$ (see \cite[Lemma 2.4]{norton}), it is possible to compute the cardinality of a linear code.
\begin{teorema} {\cite[Theorem 3.5]{norton}} \label{nele}
	A linear code $ C $ over $ R $  of type $ (k_0,\dots,k_{s-1}) $ has cardinality $ \mo{C}=p^{\sum_{i=0}^{s-1} (s-i)k_i}$. 
\end{teorema}
We attach the standard inner product to the ambient space i.e. \\$v\cdot w=\sum v_iw_i$. The dual code $C^\perp$ of $C$ is defined, as in the classical case, by  $$ C^\perp=\{v\in R^n \mid v\cdot w=0 \text{ for all }w\in C\}.$$
In \cite{wood99}, Wood proved that the dual code of a code over a  Frobenius ring, and hence over a finite chain ring, is well defined (i.e. $(C^\perp)^\perp=C)$ . The dual code $C^\perp$ has the following parameters: 
\begin{teorema}
    Let $C\subseteq R^n$ be a linear code of rank $K$ and type $ (k_0,\dots,k_{s-1}) $. Then $C^\perp$ is a linear code of rank $n-k_0$ and type $(n-K,k_{s-1},\dots,k_1)$. 
\end{teorema}
As a consequence,  the dual code of a free code is again free. \\ \newline
 We call any matrix $H$ a \emph{parity-check} matrix for $C$ if its kernel is $C$. 
\begin{proposizione}{\cite[Theorem 3.10]{norton}}
	Let $ C $ be a linear code of type $(k_0,\dots,k_{s-1})$. Then $ C $ is permutation equivalent to a code having a parity-check matrix in systematic form:  	\begin{equation}\label{pcm}
	H=\begin{bmatrix}
	B_{0,s}&B_{0,s-1}  &\dots& B_{0,1}& I_{n-K} \\
	
	\gamma B_{1,s}&\gamma B_{1,s-1} &\dots &\gamma I_{k_{s-1}}&0 \\

	\vdots &  \vdots&  & \vdots & \vdots \\
	\gamma^{s-1} B_{s-1,s} &\gamma^{s-1}I_{k_{1}}&  \dots& 0&0\\
	\end{bmatrix}=
	\begin{bmatrix}
	 H^{(0)}\\
	
	\gamma H^{(1)} \\
	\vdots\\
	\gamma^{s-1}H^{(s-1)}

	\end{bmatrix},
	\end{equation} where, for $ 0\le i,j\le s, \ B_{i,j}=-\sum_{k=i+1}^{j-1} B_{i,k}A^T_{s-j,s-k}-A^T_ {s-j,s-i}\ . $
\end{proposizione}
Clearly, $H$ is a generator matrix for $C^\perp$.
\subsection{Hamming weight distribution and Singleton-like bounds}
As in the classical case, the \emph{Hamming-weight distribution} of a ring-linear code is a vector $A = (A_i)_{i=0,1\dots,n}$,  where $A_i$ denotes the number of codewords of $C$ of weight $i$. The weight distribution
can be encoded as coefficients in a  polynomial. \begin{definizione}
The (Hamming)-\emph{weight enumerator polynomial} of a ring-linear code $C$ of length $n$ is the bivariate polynomial $$ W_C(X,Y)=\sum_{c\in C} X^{n-w(c)}Y^{w(c)}=\sum_{i=0}^n A_iX^{n-i}Y^i\ .$$

\end{definizione} 

The Hamming-weight enumerators of a code and its dual are related by the MacWilliams identities. \begin{teorema}\cite[Theorem 8.3]{wood99}\label{macwid}
    For linear codes over a finite chain ring $R$ with $p^s$ elements, the MacWilliams identities hold: \[W_{C^\perp}(X,Y)=\frac{1}{\mo{C}}W_{C}(X+(p^s-1)Y,X-Y)\ .\]
\end{teorema}

Moreover, for the Hamming metric over $R$, Singleton-like bounds are known.
\begin{rem}\label{mdspr}
The Singleton  bound for codes over any alphabet of size $p^s$ states that \begin{equation*}
     d\le n-\log_{p^s}(\mo{C})+1\ ,
\end{equation*}
(see for example \cite{macw}). In the framework of codes over finite chain rings, only free codes meet this bound and they are said \emph{maximum distance separable} (MDS) codes. \\
 As shown in \cite{dougherty}, for codes over principal ideal rings \begin{equation}\label{generalizedSing}
     d\le n-K+1 \ .
 \end{equation} This  bound is in general tighter than the Singleton bound and they coincides if and only if the code is free. A linear code over a finite chain ring meeting this bound is said to be \emph{maximum distance with respect to rank } (MDR). In particular a  code $C$ is MDS if and only if it is MDR and free.  \end{rem}
 It is well-known (\cite[Corollary 1]{shiromoto}) that the MDS property is invariant under duality (i.e. the dual of an MDS code over a finite chain ring is again MDS). In general, the dual code of an MDR code does not preserve the property. \begin{ex}
 Let $C=\langle (1,0,1),(0,2,0), (0,0,2)\rangle\subset(\Z/4\Z)^3$ be a linear code. $C$ is MDR since $d(C)=1=n-K+1$. However, its dual code $C^\perp=\langle(2,0,2), (0,2,0)\rangle$ has minimum distance $d(C^\perp)=1<2=(n-K+1)$.
 \end{ex} 
 As in the classical case of linear codes over finite fields, we  can measure how far away a linear code $C$ is  from being MDR.
\begin{definizione}
Let $C$ be a linear code of length $n$ and rank  $K$. The \emph{defect} $s(C)$ of $C$ is defined as $s(C)\coloneqq n+1-K-d$.
\end{definizione}

\subsection{On the parity-check matrix and its submatrices}
The parity-check matrix of a code give important information on some structural properties of the code, such as the minimum distance.\\
Analogously to linear codes over finite fields (see \cite[Theorem 1.4.13]{pless} ), given a code $C$ over a finite chain ring there is a link between the weights of the codewords of $ C $  and its parity-check matrix $ H $. \begin{teorema}
	Let $C$ be a linear code over $ R $ with parity-check matrix $H$. If $c\in C$, the columns
	of $H$ corresponding to the non-zero coordinates of $c$ are linearly dependent. Conversely,
	if a linear dependence relation with only non-zero coefficients exists among $w$ columns of $H$,
	then there is a codeword in $C$ of weight $w$ whose non-zero coordinates correspond to these
	columns.
\end{teorema} \begin{proof}
If $ c\in C $,  the matrix product $0=Hc^T=\sum_{i=1}^n \mathbf{h_i}c_i$, where  $\mathbf{h_i}$ is the $ i^{th} $ column of $ H $,   is a linear combination of the columns of $ H$ with coefficients provided by $c$. Conversely, if there are $ w $ linearly dependent columns in  $H$, then $ \sum_{i=0}^n \alpha_{i}\mathbf{h_{i}}=0, \  \alpha_{i}\in R$ and $ w $ of them are non-zero . If $ c=(\alpha_1,\dots,\alpha_n) $, then $ Hc^\top=0 $ and $ \text{w}(c)=w $ and $ c $ is the desired codeword.
\end{proof}
For any $  0\le j\le s-1 $, let $H^{(j)}$ be  submatrices of $H$ defined according to \eqref{pcm}. Since the $H^{(j)}$s are all full rank, the following is immediate.
\begin{proposizione}\label{0cod}
	Let $ H $ be a parity-check matrix for the  code $ C $, and hence a generator matrix for $ C^\perp $. A message $ v\in R^{n-k_0} $ is encoded as the zero codeword in $ C^\perp $ if and only if it is of the form \[v=[
\underbrace{0,\dots,0}_{n-K}||\gamma^{s-1}\cdot\underbrace{ v_1}_{k_{s-1}}||\gamma^{s-2}\cdot\underbrace{ v_2}_{k_{s-2}}||\dots||\gamma \cdot\underbrace{ v_s}_{k_{1}}
	] \ . \]
\end{proposizione} 
\begin{nota}
Let  $ M\in M_{t\times n}(R) $ be a matrix. 
According to the notation of Proposition \ref{gmstand}, we say that $M$ is of \emph{type} $(t_0,\dots,t_{s-1})$ if $ t_i $ rows of $M$ are divisible by $ \gamma^i $ but not by $ \gamma^{i+1} $,  $ 0\le i\le s-1 $.
\end{nota}

\begin{definizione}
Let $ M\in M_{t\times n}(R) $ be a matrix of type $(t_0,\dots,t_{s-1})$.\begin{itemize}
\item  For any subset $ \mathcal{I}\subseteq \{1,\dots,n\} $ of size $ \nu $, $ \mathcal{I}=\{i_1,\dots i_\nu \} $ with\\  $ i_1<i_2<\dots<i_\nu  $, we define $ \mat{M}{\I} $ as the $ t\times\nu $ submatrix of $ M $ identified by the columns indices $ \I. $
\item We define $\nmat{M}{\nu}{r_0,r_1,\dots,r_t}  $ to be the number of $t\times \nu$ submatrices of $ M $  of type $( r_0,r_1,\dots,r_t) $
\end{itemize}
\end{definizione} 
Let $C$ be an $R$-linear code of type $( k_0,k_1,\dots,k_{s-1})$, and let $ H\in M_{(n-k_0)\times n} (R)$ be a parity-check matrix in standard form for $C$. For any fixed  $ \I $ of size $\nu$, $ \mat{H}{\I} $  is a $( n-k_0)\times\nu $ submatrix of $ H $ and, as in \eqref{pcm}, we can write: 	\begin{equation}\label{HIdef}
\mat{H}{\I}=
\begin{bmatrix}
\mat{H}{\I}^{(0)}\\

\gamma\mat{H}{\I}^{(1)} \\
\vdots\\
\gamma^{s-1}\mat{H}{\I}^{(s-1)}

\end{bmatrix}\ ,
\end{equation} where each $\mat{H}{\I}^{(j)}$ is obtained from  $H^{(j)}$ by removing the unnecessary columns. Since we are selecting $\nu< n $ columns from  $H$, the type of $\mat{H}{\I} $ and $H$ may differ. First of all, note that the dimension of the first block  may decrease. Indeed: \begin{enumerate}
	\item Some rows of $ \mat{H}{\I}^{(0)} $ can be written  as  linear combinations of the others, so they can  be removed  from the parity-check matrix; 
	\item Some rows of $ \mat{H}{\I}^{(0)} $ are multiples of $ \gamma^l $ for some $1\le l\le s-1$. If  this is the case, the rows can be moved in one of the subsequent blocks. 
\end{enumerate}
On the other hand,  the dimension of the second block can either increase, decrease or  remain unchanged. One or more of the following may occur:   \begin{enumerate}
	\item Some rows of  $ \mat{H}{\I}^{(0)} $ are added to $ 	\mat{H}{\I}^{(1)} $;
	\item Some rows of of $ \mat{H}{\I}^{(1)} $ are multiples of $ \gamma^l, \ 2\le l\le s-1 $. In this case, collecting $\gamma^l$, the row  can be moved in one of the subsequent blocks;
	\item Some rows of  $ \mat{H}{\I}^{(1)} $ are linear combinations of the others, and thus they can be removed from the matrix.  
\end{enumerate} 
The procedure can be iterated for any other block  $ \mat{H}{\I}^{(j)} $, $ 2\le j\le s-1 $. %Notice that if a row in the last block $ \mat{H}{\I}^{(s-1)} $ become a multiple of $ \gamma $, the corresponding  row in $ \mat{H}{\I} $ is  a multiple of $ \gamma^s$:  this row is actually zero and it can be removed from $ \mat{H}{\I} $.
Therefore, the type of $ \mat{H}{\I} $,  being different from the type of $H$,  can not be studied in its full generality. However, if $ \nu $ satisfies certain conditions, the structure of  $ \mat{H}{\I} $ become more clear: we will focus on this problem in the next section.

%\section{Weight distribution for free  codes}\label{Freewd}
%In this section
%First of all, we focus on the case of free codes.
%We  start by studying the structure of the matrix $\mat{H}{\I}$ defined in \eqref{HIdef}.

%Free linear codes over finite chain rings have sparked the greatest attention, even from a cryptographical perspective.  %(MA è vero?)

 %This formula is well known and can be verified by a direct application of Proposition \ref{wdfree} and Theorem \ref{wdforfree}:

\section{Weight distribution for linear codes}\label{Linearwd}
From now on, let $ C $ be a linear code over a finite chain ring $ R $ of length $n$, rank $K$ and type $(k_0,k_1,\dots,k_{s-1})$.

% \begin{lemma}\label{rangofree}
% Let $C$ be a free code of rank $k_0=K$ over $R$ and let  $ H\in M_{n-K\times n}(R) $ be  a parity-check matrix for $C$ in standard form. If $n-d^\perp<\nu\le n $, then all the $(n-K)\times \nu $ submatrices of $ H $ are of type $(n-K,0,\dots,0)$.
% \end{lemma}
% \begin{proof} Given $\mathcal{I}=\{i_1,\dots i_\nu \} $, let $\mat{H}{\I}$ be an $(n-K)\times \nu$ submatrix of $H$. Notice that, since $C$ is free, the matrix $H$ is of type $(n-K,0,\dots,0)$. By contradiction suppose  $ \mat{H}{\I}$ is of type $\{r_0,r_1,\dots,r_{s-1}\}$ with $r_0<n-K$.
%   Then there exists a non-zero vector  $ v\in R^{n-K}  $ such that $ v \cdot \mat{H}{\I}=0$. By Proposition \ref{0cod}, $c\coloneqq v \cdot H  $ is a non-zero codeword in $ C^\perp $ of weight $ \w{c}\le n-\nu $ which contradicts the hypothesis $ d^\perp>n-\nu. $
% \end{proof}
\begin{lemma}\label{ranghi}
	Let $ H\in M_{(n-k_0)\times n}(R) $ be a parity-check matrix for the code $ C $. If $n-d^\perp<\nu\le n $, then  all the $ (n-k_0)\times \nu $ submatrices of $ H $ are of type $( n-K,k_{s-1}, k_{s-2},\dots, k_1 )$. In particular they all have rank $ n-K $.
\end{lemma}
\begin{proof}
%Let $ \mat{H}{\I}  $ be a submatrix of $ H $. We claim that for any choice of $ \I=\{i_1,\dots,i_\nu \}  $ the rank of $ \mat{H}{\I} $ is still $ r $.  Since the rank of $ \mat{H}{\I} $ correspond to the rank of  first block, we focus on the submatrix  $\mat{H^{(0)}}{\I}$. By contradiction assume  $  \rk{\mat{H^{(0)}}{\I}}<r$. Then there exists a non-zero vector $ \bar{v}\in R^r  $ such that $ \bar{v} \cdot \mat{H^{(0)}}{\I}=0$. Let  $ [\bar{v}|| w]=:v \in R^{n-k_0} $; by \ref{0cod}  $c\coloneqq v \cdot H  $ is a non-zero codeword in $ C^\perp $ and $ \w{c}\le n-\nu $ which contradicts the hypothesis $ d^\perp>n-\nu. $ 
Being the parity-check matrix of a linear code of type $(k_0,k_1,\dots,k_{s-1})$, $H$  is of type $( n-K,k_{s-1}, k_{s-2},\dots, k_1 )$. Let $\mat{H}{\I}$, $ \I=\{i_1,\dots, i_\nu\}$, be a $(n-k_0)\times \nu$ submatrix of $H$. Without loss of generality $\mat{H}{\I}$ can be written as in \eqref{HIdef}.
% Let $ \mat{H}{\I}  $ be a $ (n-k_0)\times\nu$ submatrix of $ H $. In  order to show that  its rank is still $n-K$, we work only  on the submatrix $ \mat{H^{(0)}}{\I} $ ( all  but the first $ n-K $ rows  are linearly dependent).  In particular, no row of $  \mat{H^{(0)}}{\I}$ is  multiple of $ \gamma $. As a consequence in $ \mat{H}{\I} $ there are at most $ k_{s-1} $ rows which are divisible by $ \gamma $ but not by $ \gamma^2 $. We claim they are exactly $ k_{s-1} $. 
By contradiction, assume that a row in $ \mat{H^{(0)}}{\I} $ is multiple of $ \gamma $. Without loss of generality we may assume it is the first one. If $ \bar{v}=[\gamma, 0,\dots,0]\in R^{n-K}  $, then $ \bar{v}\cdot\mat{H^{(1)}}{\I}=0$.  % Let $ v=[\underline{0}|| \bar{v}|| \gamma^{s-2}v_2||\dots||\gamma v_{s-1}]$, where $\underline{0}$ is the zero vector in $ R^r $ and $  v_i \in R^{k_i+1}$. 
 The vector $\bar{v}$ can be used as a first brick for  constructing a new vector $v$. By Proposition \ref{0cod}, $v$ is not encoded as the zero codeword. Therefore $ c\coloneqq v\cdot H$ is a codeword in $ C^\perp $ of weight $ w(c)\le n-\nu $, contradicting the hypothesis $ d^\perp>n-\nu$.
 So, since no row in $\mat{H^{(0)}}{\I}$ is a multiple of $\gamma$, in $ \mat{H}{\I} $ there are at most $ k_{s-1} $ rows which are multiple of  $ \gamma $ but not of $ \gamma^2 $. Applying to $\mat{H^{(1)}}{\I}$  the procedure described above, it is possible to show they are exactly $k_{s-1}$.
 Iterating the process for all the remaining blocks  $ \mat{H^{(2)}}{\I},\dots, \mat{H^{(s-1)}}{\I}$ we get the thesis. \end{proof}
\begin{cor}\label{nmatrici}
	Let $ C $ be a linear code with parity-check matrix $ H $ and $ n-d^\perp< \nu\le n. $
\[ \nmat{H}{\nu}{t_0,t_1,\dots t_{s-1}}=	\begin{cases}
		
		\binom{n}{\nu}  \text{\quad if }  t_0=n-K  \text{ and }  t_i=k_{s-i},\ 1\le i\le s-1 \\
		\hspace{5pt} 0  \quad\text{\hspace{5pt}otherwise}
	\end{cases}.\]
\end{cor}
\begin{proposizione}\label{wdlin}
Let $C $ be a code of type $(k_0,\dots, k_{s-1})$  with parity-check matrix $ H \in M_{n-k_0\times n}(R)$. Let  $ \{{A_i}\}_{i=0,\dots,n} $ be  the weight distribution of $C$. If $ n-d^\perp< \nu\le n $, then 
 \begin{equation}\label{wdist}
	\sum_{l=0}^\nu \binom{n-l}{\nu-l}A_l=\binom{n}{\nu}\frac{\mo{C}}{p^{s(n-\nu)}} \ .
	\end{equation}

\end{proposizione}
\begin{proof}
Let $ \mat{V}{\I} $ be the kernel of $ \mat{H}{\I} $.  Consider the map \[ \mat{\varphi}{\I}\colon \mat{V}{\I}\rightarrow R^n, \ \quad \mat{\varphi}{\I}(v)=(\bar{v}_1,\dots,\bar{v}_n), \quad \bar{v}_{j}=\begin{cases}
v_j \text{ if }j\in \{1,\dots,\nu\}\\ 0 \hspace{5pt}\text{ otherwise}
\end{cases}.  \] $\mat{\varphi}{\I}$ is the restriction of the map \[ \varphi_\nu\colon\bigsqcup_{\I\colon\mo{\I}=\nu} \mat{V}{\I}\to C \] to $ \mat{V}{\I} $.
The elements in the domain of $ \varphi_\nu $ can be counted in two different ways:\begin{enumerate}
\item  %Since $ n-d^\perp<\nu\le n $, by \ref{rangofree}, all the submatrices $ \mat{H}{\I}, \I=\{i_1,\dots,i_\nu\} $  have type$ \{n-k_0\} $.
For any choice of  $\I=\{i_1,\dots,i_\nu\}$, by Lemma \ref{ranghi}, the matrix $ \mat{H}{\I} $  is of length $\nu$ and of type $( n-K,k_{s-1}, k_{s-2},\dots, k_1 )$.
Therefore $ \mat{H}{\I} $ is a parity-check matrix of a linear code $C'$ of type $(\nu-n+k_0,k_1,\dots,k_{s-1})$  and length $ \nu $. Thus, by Theorem \ref{nele}, $C'$ has with $ p^{s(\nu-n+k_0)+\sum_{i=1}^{s-1} (s-i)k_i} $ elements. Hence, by  Corollary \ref{nmatrici}, 
\begin{equation}
\begin{split} \label{1}
	\left| \bigsqcup_{\I\colon\mo{\I}=\nu}\mat{V}{\I}\right|&=\nmat{H}{\nu}{n-K,k_{s-1}, k_{s-2},\dots, k_1 }\mo{\mat{V}{\I}}=\\&=\binom{n}{\nu}p^{s(\nu-n+k_0)+\sum_{i=1}^{s-1} (s-i)k_i} \ .
\end{split}    
\end{equation}
\item We consider a codeword $ c\in C $ of weight $ l\le \nu. $ Let $ \I_1=\text{supp}(c) $. Any choice of $ \nu-l $ indices $ \I_2\subset\{1,\dots,n\}\smallsetminus\I_1 $ identifies uniquely an element in $ \sqcup_{\I\colon\mo{\I}=\nu}\mat{V}{\I}. $ More precisely, $ \I_1\cup\I_2 $ determines uniquely $ \mat{H}{\I_1\cup\I_2} $, clearly $ \mat{c}{\I_1\cup\I_2}\in \mat{V}{\I_1\cup\I_2}, $ and so there is an unique element $ v\in\mat{V}{\I_1\cup\I_2}  $ such that $ \mat{\varphi}{\I_1\cup\I_2}(v)=c  $, that is $ v=\mat{c}{\I_1\cup\I_2} $. In order to determine the size of $ \varphi^{-1}_\nu(c) $, the fiber of $ c $ under the map $ \varphi_\nu $, it is enough to count all possible subsets of  $ \{1,\dots,n\}\smallsetminus\I_1 $ with size $ \nu-l$. It follows that the fiber of each codeword of weight $ l $ has $ \binom{n-l}{\nu-l} $ elements, and we observe that all the fibers of  such codewords form a partition of $ \sqcup_{\I\colon\mo{\I}=\nu} \mat{V}{\I} $. Since there are $ A_l $ codeword of weight $ l $, we obtain \begin{equation}\label{2}
\left| \bigsqcup_{\I\colon\mo{\I}=\nu}\mat{V}{\I}\right|= \sum_{l=0}^\nu \binom{n-l}{\nu-l}A_l \ .
\end{equation}
\end{enumerate}
Putting together \eqref{1} and \eqref{2} we get \eqref{wdist}.
\end{proof}
\begin{cor}
    For a free code  $C$ of length $n$ and rank $k_0=K$ over $R$ the weight distribution formula reads: \begin{equation*}
      \sum_{l=0}^\nu \binom{n-l}{\nu-l}A_l=\binom{n}{\nu}p^{s(\nu+K-n)} \ .  
    \end{equation*}
\end{cor}

\begin{teorema}\label{optimal}
    Let $\sigma$ be the sum of the Singleton  defects of $C$ and $C^\perp$. the knowledge of $\sigma+d+K-k_0-1$ elements of the weight distribution $\{A_0,\dots, A_n\}$ is enough to compute the full weight distribution of $C$ and $C^\perp$. In particular the knowledge of $d$ and of any $\sigma+K-k_0-1$ elements of $\{A_d,\dots, A_n\}$ is enough to compute the entire weight distribution of $C$ and $C^\perp$.  
\end{teorema}

\begin{proof} Consider equation \eqref{wdist} with $\nu$ varying  in range  $\{n-d^\perp+1,\dots,n\}$. 
We obtain a linear system of the form $$\mathcal{P}\cdot A(C)=b \ ,$$ where $\mathcal{P}$ is a truncated Pascal matrix with $d^\perp$ rows and $n$ columns. By \cite{ker}, all the minors  of $\mathcal{P}$ of order $d^\perp$ are non-zero. Hence, the knowledge of $n-d^\perp+1=K-\sigma-1+d-(K-k_0)$ elements in the weight distribution allow us to obtain a linear system that admit a unique solution. Finally, the knowledge of $d$ implies the knowledge of $A_0,\dots,A_{d-1}$. Hence, it is sufficient to know other $n-d^\perp+1-d=\sigma+K-k_0-1$ elements in $\{A_d,\dots A_n\}$ to determine the full weight  distribution of the code. 
\end{proof}

We now show that, at least in some cases, Theorem \ref{optimal} is optimal. Indeed, in general, it is not possible to deterministically deduce the weight distribution of a linear code with less then $\sigma+K-k_0-1$ elements in $\{A_0,\dots,A_n\}$. 

\begin{definizione}
A linear  code is said to be \emph{Almost-MDR code}(AMDR)  if it has Singleton defect equal to 1. An AMDR  code whose dual is still  AMDR  is called \emph{Near-MDR}.
\end{definizione}
\begin{definizione}
A free code is said to be \emph{Almost-MDS code}(AMDS)  if it has  Singleton defect equal to 1. An AMDS  code whose dual is still  AMDS  is called \emph{Near-MDS} 
\end{definizione}
 Consider two Near-MDS codes with the same parameters.  For a Near-MDS  only $\sigma+K-k_0-1=1$ weight is  necessary to determine the full weight distribution. If the sole knowledge of $d$ and $n $ was enough  to compute the entire weight distribution of the code, then any two near-MDS would be formally equivalent. But this is in general false: 
\begin{ex}
Let $C_1 $ and $C_2$ be two near-MDS codes over  $\Z_{5^3}$ of length $n=4$, rank $K=2$ and minimum Hamming distance $d=2$, generated respectively  by  \begin{equation*}
    G_1= \begin{bmatrix}
1 & 0 & 57 & 0  \\
0 & 1 & 0& 68  \\
\end{bmatrix},
\end{equation*}
and
\begin{equation*}
    G_2= \begin{bmatrix}
1 & 0 & 5 & 43  \\
0 & 1 & 82  & 5 \\
\end{bmatrix} .
\end{equation*}
Their weight distributions are respectively: $$\mathcal{A}_1=(1,0,248,0,15376)\ ,$$ and $$\mathcal{A}_2=(1,0,8,480,15136) \ .$$
Therefore $C_1$ and $C_2 $ are not formally equivalent. 
\end{ex}

\iffalse
\begin{ex}
 Consider a Near-MDS code.  For a Near-MDS code $\sigma-k_0+K-1=1$. If the sole knowledge of $d$ and $n $ was enough  to compute the entire weight distribution of the code, then any two near-MDS would be formally equivalent. However, this is  false as shown in the following example. Let $C_1 $ and $C_2$ be two near-MDS codes over  $\Z_{25}$ of length $n=12$, rank $K=6$ and minimum Hamming distance $d=6$, generated  by  \begin{equation*}
    G_1= \begin{bmatrix}
1 & 0 & 0 & 0 & 0 & 0 & 19 & 12 & 1 & 1 & 5 & 8  \\
0 & 1 & 0 & 0 & 0 & 0 & 22 & 23 & 12 & 8 & 9& 12  \\
0 & 0 & 1 & 0 & 0 & 0 & 4  & 2 & 9 & 5 & 6 & 22  \\
0 & 0 & 0 & 1 & 0 & 0 & 12& 22 & 13 & 9 & 17  & 2 \\
0 & 0 & 0 & 0 & 1 & 0 & 4 & 11& 0 & 13 & 21 & 18  \\
0 & 0 & 0 & 0 & 0 & 1 & 0  & 17 & 24 & 9 & 18 & 1
\end{bmatrix},
\end{equation*}
and
\begin{equation*}
    G_2= \begin{bmatrix}
1 & 0 & 0 & 0 & 0 & 0 & 11 & 16 & 20 & 9 & 18 & 8  \\
0 & 1 & 0 & 0 & 0 & 0 & 18& 11 & 1 & 22 & 30  & 21 \\
0 & 0 & 1 & 0 & 0 & 0 & 4& 12 & 5& 5 & 1 & 12  \\
0 & 0 & 0 & 1 & 0 & 0 & 17  & 24 & 22& 24 & 2 & 0\\
0 & 0 & 0 & 0 & 1 & 0 &14 & 5& 23 & 18 &14 & 24\\
0 & 0 & 0 & 0 & 0 & 1 & 0& 7 & 4 & 9 & 23  & 11 
\end{bmatrix} .
\end{equation*}
The weight distributions of $C_1$ and $C_2$ are respectively: $$A_1=(1,0,0,0,0,0,1000,13008,228840,2371840,17147544,74790480,149587912)$$ and $$A_2=(1,0,0,0,0,0,960,13248,228240,2372640,17146944,
149587872);$$
therefore $C_1$ and $C_2 $ are not formally equivalent. 
\end{ex}\fi

\section{Weight distribution of codes with small Singleton defects}\label{appwdr}

Theorem \ref{optimal} lead us to further investigate codes having a small number of Singleton defects.
\\ \newline We start focusing on codes meeting  the Singleton bound: from Remark \ref{mdspr} follows that  $\sigma=0$  for any MDS code. \\ 
 The weight distribution of MDS codes is well known (see \cite[ Theorem 5]{shiromoto}); however   it can be directly obtained with the  sole knowledge of the length  and  the minimum distance  from  Theorem \ref{optimal}: 
\begin{teorema}
 Let $C$ be a MDS code of length $n$ and rank $k$, then 
 \[ A_w(C)=\binom{n}{w}\sum_{j=0}^{w-d} (-1)^j \binom{i}{j} (p^{s-d+1-j}-1)\ .\]
\end{teorema}
Now we move to  linear codes meeting the generalized Singleton bound \eqref{generalizedSing}, the  MDR codes. As  shown in \ref{mdspr} the dual of an MDR code is not necessarily MDR. Therefore the weight distribution may depend on one or more parameters. \\
Let  $C$ be an MDR code of length $n$, rank $K$, minimum distance $d=n-K+1$   and let  $C^\perp$  be its dual having rank $n-K$ and minimum distance $k_0-\sigma+1$ for some $\sigma\ge 0$.
% Let $p$ be the cardinality of  the base field  $\KK$ \\ 
 According to the notation  of Proposition \ref{wdist}, since $\nu<n-d^\perp$ and $d^\perp=k_0-\sigma+1$, we can write $\nu=n-k_0+\sigma+i$ with  $ i$ ranging in $\{0,\dots, k_0-\sigma\}.$ Let $q=\frac{\mo{C}}{p^{s(n-\nu)}}$. Then equation \eqref{wdist} becomes \[\sum_{l=0}^{n-k_0+\sigma-1+i}\binom{n-l}{n-k_0+\sigma+i-l}A_l=\binom{n}{n-k_0+\sigma+i}q \ . \] Since $A_0=1 $ and $A_l=0$ for all $1\le l \le n-k$, we may write
    \begin{align*}
   \binom{n}{n-k_0+\sigma+i}+\sum_{h=0}^{\sigma-k+k_0-2} \binom{k-1-h}{\sigma+i+K-k_0-h} A_{n+K+1+h}+&\\ +\sum_{l=n+\sigma-k_0}^{n-k_0+\sigma+i}\binom{n-l}{n-k_0+\sigma+i-l}A_l=\binom{n}{n-k_0+\sigma+i}q \ . 
\end{align*}
Let $l=n+\sigma-k_0+j$, then 
\begin{align*}
    &\sum_{j=0}^i \binom{k_0-\sigma-j}{i-j} A_{n+\sigma-k_0+j}=\\=\binom{n}{n-k_0+\sigma+i}&(q-1)-\sum_{h=0}^{\sigma+K-k_0-2}\binom{K-1-h}{K-k_0+\sigma+i+h}A_{n+j+1+h} \ .
\end{align*}
%\end{equation*}
We can re-write the linear system in matrix form, as $\mathcal{P}\cdot A=b$ with $\mathcal{P}$  Pascal matrix $\Big[\binom{k_0-\sigma-j}{i-j}\Big]_{i,j=0,\dots,k_0-\sigma}$. Therefore $\mathcal{P}^{-1}=(-1)^{i-j}\Big[\binom{k_0-\sigma-j}{i-j}\Big]_{i,k=0,\dots,k_0-\sigma}$. More explicitly: \begin{proposizione}
  Let  $C$ be an MDR code of length $n$, rank $K$, minimum distance $d=n-K+1$   and let  $C^\perp$  be its dual having rank $n-k_0$ and minimum distance $k_0-\sigma+1$ for some $\sigma\ge 0$. Let $q=\frac{\mo{C}}{p^{s(n-\nu)}}$.  The knowledge of $\{A_{n-K+1},\dots, A_{n-k_0+\sigma-1} \}$ in the weight distribution of $C$ is enough to compute the entire weight distribution of $C$.\\  In particular, for all $0\le i \le k_0-\sigma-K$, we have:
   \begin{align*}
       A_{n-k_0+\sigma+i}=\sum_{j=0}^i (-1)^{i-j}&\binom{k_0-\sigma-j}{i-j}\Bigg[\binom{n}{n-k_0+\sigma+i}(q-1)-\\&-\sum_{h=0}^{\sigma+K-k_0-2}\binom{k-h-1}{k-k_0+\sigma+i-h-1}A_{n+k+1+h}\Bigg] \ .
   \end{align*}
\end{proposizione}

In a similar fashion we can derive the weight distribution of an AMDR code: 
\begin{proposizione}
     Let  $C$ be an AMDR code of length $n$, rank $K$, minimum distance $d=n-K$   and let  $C^\perp$  be its dual having rank $n-k_0$ and minimum distance $k_0-\sigma+1$ for some $\sigma\ge 0$. Let $q=\frac{\mo{C}}{p^{s(n-\nu)}}$. The knowledge of $\{A_{n-K+1},\dots, A_{n-k_0+\sigma-1} \}$ is enough to compute the entire weight distribution of $C$.\\ In particular, for all $0\le i \le k_0-\sigma-K$ we have: 
   \begin{align*}
       A_{n-k_0+\sigma+i}=\sum_{j=0}^i (-1)^{i-j}&\binom{k_0-\sigma-j+1}{i-j}\Bigg[\binom{n}{n-k_0+\sigma+i-1}(q-1)-\\&\sum_{h=0}^{\sigma+K-k_0-2}\binom{k-h}{k-k_0+\sigma+i-h+1}A_{n+k+h}\Bigg] \ .
   \end{align*}
\end{proposizione}
Clearly, by specializing the previous formula, we also get the weight distributions of Near-MDS  and Near-MDR codes.
\section{Relation with MacWilliams identities}\label{Relation with MacWilliams identities}

Both in classical and ring-linear coding theory, the most fundamental result about weight distributions are the MacWilliams identities (Theorem \ref{macwid}). They relate the weight enumerator polynomial of a linear code and its dual.  However, in our framework it is more convenient to work with other equivalent set of equations in place of the polynomial form of \ref{macwid}. 
Following the outline of \cite[Chapter 5, Section 2]{macw},   and combining it with \ref{macwid}  we can deduce the following equality:
\begin{equation*}
    \sum_{j=0}^{n-\nu} \binom{n-j}{\nu}A_j=\frac{\mo{C}}{p^{s\nu}}\sum_{j=0}^\nu \binom{n-j}{n-\nu}A^\perp_j\ ,\ \text{ for }\  0\le \nu\le n.
\end{equation*} 
Moreover, in a similar fashion to \cite[Theorem 7.2.3]{pless},  we get:
\begin{equation} \label{expmacw}
    \sum_{j=0}^n\binom{j}{\nu}A_j =\frac{\mo{C}}{p^{s\nu}}\sum_{j=0}^\nu  (-1)^{j}\binom{n-j}{n-\nu}(p^s-1)^{\nu-j}A_j^\perp\ , \ \text{ for }\ 0\le \nu\le n.
\end{equation}
If $\nu< d^\perp$, each $A_j^\perp$ of the right hand side of \eqref{expmacw} is equal to zero except for $A_0^\perp$ which is equal to 1. Therefore we get a ring-variant of Pless' equations. 
\begin{proposizione} \label{pless}
For any $\nu<d^\perp$ 
\begin{equation} \label{plesseq}
    \sum_{j=0}^n\binom{j}{\nu}A_j =\frac{\mo{C}}{p^{s\nu}}\binom{n}{n-\nu}(p^s-1)^{\nu}, \  \text{ for } \ 0\le \nu<d^\perp.
\end{equation}\end{proposizione} Hence, when enough terms of the weight distributions are known, systems in \ref{wdlin} and \ref{pless} are equivalent. 
\begin{cor}
    Let $\sigma$ be the sum of the Singleton  defects of $C$ and $C^\perp$. Using equation \eqref{plesseq}, the knowledge of $\sigma+d+K-k_0-1$ elements of the weight distribution $\{A_0,\dots, A_n\}$ is enough to compute the full weight distribution of $C$ and $C^\perp$. In particular the knowledge of $d$ and of any $\sigma+K-k_0-1$ elements of $\{A_d,\dots, A_n\}$ is enough to compute the entire weight distribution of $C$ and $C^\perp$.  
\end{cor}
\begin{proof} The proof follows the same outline of Proposition \ref{optimal}.
\end{proof}
Therefore, the two systems of equations  \eqref{wdlin} and \eqref{plesseq} are equivalent provided the existence of the code.
\section{Conclusion}
In analogy to linear codes over finite fields,  the minors of the parity-check matrix of a ring-linear code enable us to determine linear relations between the weights of the codes.
Our formulae enable to verify the weight distribution of MDS codes.  Moreover this result allows to determine the full weight distributions of MDR, Near-MDR, AMDR codes.\\
The number of parameters necessary to derive the full weight distribution of a code and its dual depends on the the sum of the Singleton defects of the code and its dual, an in particular it is bounded by $n+K-k_0-1$.\\
An interesting extension of this work would be the study of more classes of codes, either by considering the case of non-AMDR codes or families obtained via structured parity check matrices. 
A second promising line of research would be the derivation of formulas for weight distribution related to different metrics, e.g. Lee metric or Rank metric.  
\section*{Acknowledgement}
The publication was created with the co-financing of the European Union -  FSE-REACT-EU, PON Research and Innovation 2014-2020 DM1062 / 2021. The authors are member of the INdAM Research Group GNSAGA. The core of this work was partially presented on a talk given at COMBINATORICS 2022 in Mantova, Italy by the first author.
\bibliographystyle{plain}
\bibliography{bibliografia.bib}

@article{norton,
  title={On the structure of linear and cyclic codes over a finite chain ring},
  author={Norton, Graham H and S{\u{a}}l{\u{a}}gean, Ana},
  journal={Applicable algebra in engineering, communication and computing},
  volume={10},
  number={6},
  pages={489--506},
  year={2000},
  publisher={Springer}
}
@book{pless,
	title={Fundamentals of error-correcting codes},
	author={Huffman, W Cary and Pless, Vera},
	year={2010},
	publisher={Cambridge university press}
}
@article{meneghetti2021,
	title={A formula on the weight distribution of linear codes with applications to AMDS codes},
	author={Meneghetti, Alessio and Pellegrini, Marco and Sala, Massimiliano},
	journal={Finite Fields and Their Applications},
	volume={77},
	pages={101933},
	year={2022},
	publisher={Elsevier}
}
@article{wood99,
  title={Duality for modules over finite rings and applications to coding theory},
  author={Wood, Jay A},
  journal={American journal of Mathematics},
  pages={555--575},
  year={1999},
  publisher={JSTOR}
}
@article{shiromoto,
  title={Note on MDS codes over the integers modulo $p^m$},
  author={SHIROMOTO, Keisuke},
  journal={Hokkaido Mathematical Journal},
  volume={29},
  number={1},
  pages={149--157},
  year={2000},
  publisher={Hokkaido University, Department of Mathematics}
}

@article{dougherty,
  title={MDR codes over $\mathbb{Z}_k$},
  author={Dougherty, Steven T and Shiromoto, Keisuke},
  journal={IEEE Transactions on Information Theory},
  volume={46},
  number={1},
  pages={265--269},
  year={2000},
  publisher={IEEE}
}

@article{vardy,
  title={The intractability of computing the minimum distance of a code},
  author={Vardy, Alexander},
  journal={IEEE Transactions on Information Theory},
  volume={43},
  number={6},
  pages={1757--1766},
  year={1997},
  publisher={IEEE}
}

@book{macw,
  title={The theory of error correcting codes},
  author={MacWilliams, Florence Jessie and Sloane, Neil James Alexander},
  volume={16},
  year={1977},
  publisher={Elsevier}
}

@article{ker,
  title={Invertibility of submatrices of Pascal's matrix and Birkhoff interpolation},
  author={Kersey, Scott N},
  journal={arXiv preprint arXiv:1303.6159},
  year={2013}
}

@article{kerdock,
  title={The Z4-linearity of Kerdock, Preparata, Goethals and related codes},
  author={Calderbank, AR and Hammons Jr, AR and Kumar, P Vijay and Sloane, NJA and Sol{\'e}, P},
  journal={IEEE Trans. Inf. Theory},
  volume={40},
  number={2},
  pages={301--319},
  year={1994}
}

@article{horlemann,
  title={Information set decoding in the Lee metric with applications to cryptography},
  author={Horlemann-Trautmann, Anna-Lena and Weger, Violetta},
  journal={arXiv preprint arXiv:1903.07692},
  year={2019}
}

@article{weger,
  title={On the hardness of the Lee syndrome decoding problem},
  author={Weger, Violetta and Khathuria, Karan and Horlemann, Anna-Lena and Battaglioni, Massimo and Santini, Paolo and Persichetti, Edoardo},
  journal={arXiv preprint arXiv:2002.12785},
  year={2020}
}


@article{persichetti,
  title={Information set decoding of Lee-metric codes over finite rings},
  author={Weger, Violetta and Battaglioni, Massimo and Santini, Paolo and Chiaraluce, Franco and Baldi, Marco and Persichetti, Edoardo},
  journal={arXiv preprint arXiv:2001.08425},
  year={2020}
}


@article{mceliece,
  title={A public-key cryptosystem based on algebraic},
  author={McEliece, Robert J},
  journal={Coding Thv},
  volume={4244},
  pages={114--116},
  year={1978}
}

@book{mcw,
  title={The theory of error correcting codes},
  author={MacWilliams, Florence Jessie and Sloane, Neil James Alexander},
  volume={16},
  year={1977},
  publisher={Elsevier}
}

@misc{torleiv,
  title={Codes for Error Detection, Serial on Coding Theory and Cryptography, vol. 2},
  author={Torleiv, Klove},
  year={2007},
  publisher={World Scientific, Singapore}
}

@article{meneghetti2022,
  title={On the equivalence of two post-quantum cryptographic families},
  author={Meneghetti, Alessio and Pellegrini, Alex and Sala, Massimiliano},
  journal={Annali di Matematica Pura ed Applicata (1923-)},
  pages={1--25},
  year={2022},
  publisher={Springer}
}

@article{crystals,
  title={CRYSTALS-Kyber algorithm specifications and supporting documentation},
  author={Avanzi, Roberto and Bos, Joppe and Ducas, L{\'e}o and Kiltz, Eike and Lepoint, Tancr{\`e}de and Lyubashevsky, Vadim and Schanck, John M and Schwabe, Peter and Seiler, Gregor and Stehl{\'e}, Damien},
  journal={NIST PQC Round},
  volume={2},
  number={4},
  pages={1--43},
  year={2017}
}

@article{bike,
  title={BIKE: bit flipping key encapsulation},
  author={Aragon, Nicolas and Barreto, Paulo SLM and Bettaieb, Slim and Bidoux, Loic and Blazy, Olivier and Deneuville, Jean-Christophe and Gaborit, Philippe and Gueron, Shay and Guneysu, Tim and Melchor, Carlos Aguilar and others},
  year={2017}
}

@incollection{bernstein,
  title={Introduction to post-quantum cryptography},
  author={Bernstein, Daniel J},
  booktitle={Post-quantum cryptography},
  pages={1--14},
  year={2009},
  publisher={Springer}
}

@article{nist,
  title={Status report on the second round of the NIST post-quantum cryptography standardization process},
  author={Moody, Dustin and Alagic, Gorjan and Apon, Daniel C and Cooper, David A and Dang, Quynh H and Kelsey, John M and Liu, Yi-Kai and Miller, Carl A and Peralta, Rene C and Perlner, Ray A and others},
  year={2020}
}

@article{hqc,
  title={Hamming quasi-cyclic (HQC)},
  author={Melchor, Carlos Aguilar and Aragon, Nicolas and Bettaieb, Slim and Bidoux, Lo{\i}c and Blazy, Olivier and Deneuville, Jean-Christophe and Gaborit, Philippe and Persichetti, Edoardo and Z{\'e}mor, Gilles and Bourges, IC},
  journal={NIST PQC Round},
  volume={2},
  number={4},
  pages={13},
  year={2018}
}

@article{sike,
  title={Supersingular isogeny key encapsulation},
  author={Azarderakhsh, Reza and Campagna, Matthew and Costello, Craig and Feo, LD and Hess, Basil and Jalali, Amir and Jao, David and Koziel, Brian and LaMacchia, Brian and Longa, Patrick and others},
  journal={Submission to the NIST Post-Quantum Standardization project},
  volume={152},
  pages={154--155},
  year={2017}
}

@article{berlekamp,
  title={On the inherent intractability of certain coding problems (corresp.)},
  author={Berlekamp, Elwyn and McEliece, Robert and Van Tilborg, Henk},
  journal={IEEE Transactions on Information Theory},
  volume={24},
  number={3},
  pages={384--386},
  year={1978},
  publisher={IEEE}
}

@article{pellegrini,
  title={Weight distribution of Hermitian codes and matrices rank},
  author={Pellegrini, Marco and Sala, Massimiliano},
  journal={Finite Fields and Their Applications},
  volume={60},
  pages={101578},
  year={2019},
  publisher={Elsevier}
}

\begin{thebibliography}{10}

\bibitem{bike}
Nicolas Aragon, Paulo~SLM Barreto, Slim Bettaieb, Loic Bidoux, Olivier Blazy,
  Jean-Christophe Deneuville, Philippe Gaborit, Shay Gueron, Tim Guneysu,
  Carlos~Aguilar Melchor, et~al.
\newblock Bike: bit flipping key encapsulation.
\newblock 2017.

\bibitem{crystals}
Roberto Avanzi, Joppe Bos, L{\'e}o Ducas, Eike Kiltz, Tancr{\`e}de Lepoint,
  Vadim Lyubashevsky, John~M Schanck, Peter Schwabe, Gregor Seiler, and Damien
  Stehl{\'e}.
\newblock Crystals-kyber algorithm specifications and supporting documentation.
\newblock {\em NIST PQC Round}, 2(4):1--43, 2017.

\bibitem{berlekamp}
Elwyn Berlekamp, Robert McEliece, and Henk Van~Tilborg.
\newblock On the inherent intractability of certain coding problems (corresp.).
\newblock {\em IEEE Transactions on Information Theory}, 24(3):384--386, 1978.

\bibitem{bernstein}
Daniel~J Bernstein.
\newblock Introduction to post-quantum cryptography.
\newblock In {\em Post-quantum cryptography}, pages 1--14. Springer, 2009.

\bibitem{kerdock}
AR~Calderbank, AR~Hammons~Jr, P~Vijay Kumar, NJA Sloane, and P~Sol{\'e}.
\newblock The z4-linearity of kerdock, preparata, goethals and related codes.
\newblock {\em IEEE Trans. Inf. Theory}, 40(2):301--319, 1994.

\bibitem{dougherty}
Steven~T Dougherty and Keisuke Shiromoto.
\newblock Mdr codes over $\mathbb{Z}_k$.
\newblock {\em IEEE Transactions on Information Theory}, 46(1):265--269, 2000.

\bibitem{horlemann}
Anna-Lena Horlemann-Trautmann and Violetta Weger.
\newblock Information set decoding in the lee metric with applications to
  cryptography.
\newblock {\em arXiv preprint arXiv:1903.07692}, 2019.

\bibitem{pless}
W~Cary Huffman and Vera Pless.
\newblock {\em Fundamentals of error-correcting codes}.
\newblock Cambridge university press, 2010.

\bibitem{ker}
Scott~N Kersey.
\newblock Invertibility of submatrices of pascal's matrix and birkhoff
  interpolation.
\newblock {\em arXiv preprint arXiv:1303.6159}, 2013.

\bibitem{macw}
Florence~Jessie MacWilliams and Neil James~Alexander Sloane.
\newblock {\em The theory of error correcting codes}, volume~16.
\newblock Elsevier, 1977.

\bibitem{mceliece}
Robert~J McEliece.
\newblock A public-key cryptosystem based on algebraic.
\newblock {\em Coding Thv}, 4244:114--116, 1978.

\bibitem{hqc}
Carlos~Aguilar Melchor, Nicolas Aragon, Slim Bettaieb, Lo{\i}c Bidoux, Olivier
  Blazy, Jean-Christophe Deneuville, Philippe Gaborit, Edoardo Persichetti,
  Gilles Z{\'e}mor, and IC~Bourges.
\newblock Hamming quasi-cyclic (hqc).
\newblock {\em NIST PQC Round}, 2(4):13, 2018.

\bibitem{meneghetti2022}
Alessio Meneghetti, Alex Pellegrini, and Massimiliano Sala.
\newblock On the equivalence of two post-quantum cryptographic families.
\newblock {\em Annali di Matematica Pura ed Applicata (1923-)}, pages 1--25,
  2022.

\bibitem{meneghetti2021}
Alessio Meneghetti, Marco Pellegrini, and Massimiliano Sala.
\newblock A formula on the weight distribution of linear codes with
  applications to amds codes.
\newblock {\em Finite Fields and Their Applications}, 77:101933, 2022.

\bibitem{nist}
Dustin Moody, Gorjan Alagic, Daniel~C Apon, David~A Cooper, Quynh~H Dang,
  John~M Kelsey, Yi-Kai Liu, Carl~A Miller, Rene~C Peralta, Ray~A Perlner,
  et~al.
\newblock Status report on the second round of the nist post-quantum
  cryptography standardization process.
\newblock 2020.

\bibitem{norton}
Graham~H Norton and Ana S{\u{a}}l{\u{a}}gean.
\newblock On the structure of linear and cyclic codes over a finite chain ring.
\newblock {\em Applicable algebra in engineering, communication and computing},
  10(6):489--506, 2000.

\bibitem{pellegrini}
Marco Pellegrini and Massimiliano Sala.
\newblock Weight distribution of hermitian codes and matrices rank.
\newblock {\em Finite Fields and Their Applications}, 60:101578, 2019.

\bibitem{shiromoto}
Keisuke SHIROMOTO.
\newblock Note on mds codes over the integers modulo $p^m$.
\newblock {\em Hokkaido Mathematical Journal}, 29(1):149--157, 2000.

\bibitem{torleiv}
Klove Torleiv.
\newblock Codes for error detection, serial on coding theory and cryptography,
  vol. 2, 2007.

\bibitem{vardy}
Alexander Vardy.
\newblock The intractability of computing the minimum distance of a code.
\newblock {\em IEEE Transactions on Information Theory}, 43(6):1757--1766,
  1997.

\bibitem{persichetti}
Violetta Weger, Massimo Battaglioni, Paolo Santini, Franco Chiaraluce, Marco
  Baldi, and Edoardo Persichetti.
\newblock Information set decoding of lee-metric codes over finite rings.
\newblock {\em arXiv preprint arXiv:2001.08425}, 2020.

\bibitem{weger}
Violetta Weger, Karan Khathuria, Anna-Lena Horlemann, Massimo Battaglioni,
  Paolo Santini, and Edoardo Persichetti.
\newblock On the hardness of the lee syndrome decoding problem.
\newblock {\em arXiv preprint arXiv:2002.12785}, 2020.

\bibitem{wood99}
Jay~A Wood.
\newblock Duality for modules over finite rings and applications to coding
  theory.
\newblock {\em American journal of Mathematics}, pages 555--575, 1999.

\end{thebibliography}
\end{document}